\documentclass[conference]{IEEEtran}
\IEEEoverridecommandlockouts
\usepackage{mathrsfs}
\usepackage{amsmath}
\usepackage{amssymb}
\usepackage{amsthm}
\ifCLASSINFOpdf \else \fi

\usepackage[dvips]{graphicx}
\usepackage{algorithm}
\usepackage{algorithmic}

\newtheorem{Theo}{Theorem}
\newtheorem{Lemm}{Lemma}

\usepackage[dvips]{graphicx}
\usepackage{algorithm}
\usepackage{algorithmic}

\begin{document}
\title{On the Construction of Radio Environment Maps for Cognitive Radio Networks}
\author{\IEEEauthorblockN{Zhiqing~Wei\IEEEauthorrefmark{2}, Qixun~Zhang\IEEEauthorrefmark{1}, Zhiyong~Feng\IEEEauthorrefmark{1}}
\IEEEauthorblockA{Wireless Technology Innovation Institute\\
Beijing University of Posts and Telecommunications\\
Beijing 100876, P. R. China\\
Email: \IEEEauthorrefmark{2}zhiqingwei@gmail.com, \IEEEauthorrefmark{1}\{zhangqixun, fengzy\}@bupt.edu.cn} \and
\IEEEauthorblockA{Wei Li,~and~T. Aaron Gulliver\\
Department of Electrical and Computer Engineering\\
University of Victoria\\
Victoria, BC, Canada V8W 3P6\\
Email: weili@ieee.org, a.gullive@ece.uvic.ca}
\thanks{This work is accepted by WCNC 2013, supported by the National Basic Research Program (973 Program) of China (No. 2009CB320400), the
National Natural Science Foundation of China (61227801, 61201152, 61121001), the National Key Technology R\&D
Program of China (2012ZX03003006), the Program for New Century Excellent Talents in
University (NCET-01-0259).}}
\maketitle

\begin{abstract}
The Radio Environment Map (REM) provides an effective approach to Dynamic Spectrum Access (DSA) in Cognitive Radio Networks (CRNs).
Previous results on REM construction show that there exists a tradeoff between the number of measurements (sensors) and REM accuracy.
In this paper, we analyze this tradeoff and determine that the REM error is a decreasing and convex function of the number of measurements (sensors).
The concept of geographic entropy is introduced to quantify this relationship.
And the influence of sensor deployment on REM accuracy is examined using information theory techniques.
The results obtained in this paper are applicable not only for the REM, but also for wireless sensor network deployment.\\
\end{abstract}


\IEEEpeerreviewmaketitle

\section{Introduction}

Increases in the number of wireless communication systems has created a heterogeneous radio environment
where multiple Radio Access Technologies (RATs) coexist in the same time and space.
As a result, User Equipment (UE) with cognitive capabilities is crucial for flexible radio resource usage.
Mitola first proposed Cognitive Radio (CR) in 1998 as a context-aware radio technology that can be reconfigured to adapt to the environment \cite{CR_Mitola}.
The Radio Environment Map (REM) has been proposed as a database for dynamic spectrum access based on UE location and spectrum usage.
It contains multi-dimensional cognitive information such as geographic features, spectral regulations, equipment locations, radio activity logs,
user policies, and service providers \cite{CR_Performance_Evaluation}.

To build a REM, sensors (or UE) must be deployed to detect the radio environment.
The measurement data from the sensors is reported to an REM manager.
Several approaches have been employed for REM construction.
In \cite{REM_Fast_Algorithm}, Grimoud \emph{et al.} used an iterative process to obtain the REM based on Kriging interpolation
to reduce the measurement data required.
Riihij\"{a}rvi \emph{et al.} \cite{REM_spatial_statistics} developed a probabilistic model for the REM which
exploits the correlation in the measured data to reduce the complexity.
And Atanasovski \emph{et al.} \cite{REM_Heterogeneous_sensor} produced an REM prototype using heterogeneous spectrum sensors.

The goal of previous work on REM construction was to reduce the number of measurements required and improve REM accuracy.
There is a tradeoff between the number of measurements (or number of sensors), and the accuracy.
Faint \emph{et al.} \cite{REM_number_sensor_and_REM} examined this relationship using computer simulation,
and showed that increasing the sensor density can increase REM accuracy.
However, when the sensors are sufficiently dense, the improvement is not significant.
In this paper, we examine this tradeoff theoretically and determine that the relationship between the radio parameter error (REM accuracy),
is ${p_e} = \Theta (\frac{1}{{\sqrt M }})$, where $M$ is the number of sensors.
Besides, we obtain a closed form expression for $p_e$ as a function of $M$, which is a decreasing and convex function.
This verifies the simulation results in \cite{REM_number_sensor_and_REM}.
Converse to previous approaches, we build the REM by considering the coverage of all networks,
which is inspired by the Cognitive Pilot Channel (CPC) technology in \cite{CPC_Ondemand}.
Our results are not only applicable to REM construction, but also to deployment in wireless sensor networks (WSN).

The rest of this paper is organized as follows.
Sensor deployment and its relationship to REM construction is presented in Section II.
The analysis of this relationship is provided in Section III.
In Section IV, we examine the tradeoff between the number of sensors and REM accuracy.
Section V presents some numerical results, and finally some concluding remarks are given in Section VI.

\section{Sensor Deployment}

The region is divided into small meshes, which are shown as small squares in Fig. \ref{fig_system_model}.
Sensors are deployed over the entire region, and can be network detectors, spectrum sensing entities or just UE.
Two sensor deployment schemes are considered, one-mesh-one-sensor and random sensor deployment.
In the one-mesh-one-sensor scheme, a sensor is deployed in each mesh randomly.
Thus the number of sensors is equal to the number of meshes.
A sensor measurement is considered to be the radio environment for the entire mesh.
Thus after gathering all sensor measurements, the REM can be constructed (an example is shown in Fig. \ref{fig_picture_REM_construction}($\rm{a_1}$)).
With random sensor deployment, the sensors are randomly deployed in the region without regard for mesh boundaries.
In this case, the majority of the sensor measurements in a mesh determines the radio environment, and these values are
used to construct the REM for the region (examples are given in Fig. \ref{fig_picture_REM_construction}($\rm{b_1}$) and ($\rm{c_1}$)).

\begin{figure}[!t]
\centering
\includegraphics[width=0.35\textwidth]{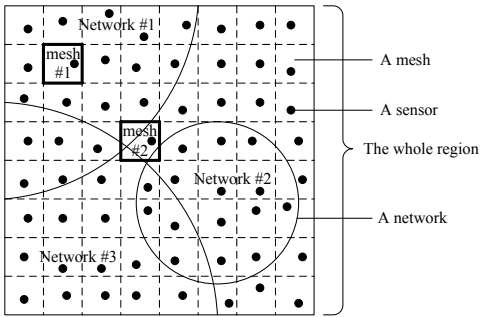}
\caption{The heterogeneous radio network distribution and sensor deployment for REM construction.} \label{fig_system_model}
\end{figure}

\section{REM Construction}

\subsection{REM parameters}

Define the binary representation of network $k$ at location $(x,y)$ as
\begin{equation}
\label{eq_binary_representation}
R{\rm{(}}k,x,y{\rm{)}} = \left\{ \begin{array}{l}
1 \mbox{ if network } k \mbox{ is detected at }(x,y)\\
0 \mbox{ otherwise}\\
\end{array} \right.
\end{equation}
Radio parameter at a location is characterized by the following sum of the binary representations for all networks \cite{CPC_Ondemand}
\begin{equation}\label{eq_radio_parameters}
I(x,y) = \sum \limits_{k = 1}^T {R(k,x,y) \times 2^{k - 1} }
\end{equation}
where $T$ is the number of networks.
The radio parameter of mesh $i$ is then
\begin{equation}\label{eq_radio_parameter_majority}
P = \mathop {\arg\max {p_{ij}}}\limits_j
\end{equation}
where $p_{ij}$ is the fraction of the area in mesh $i$ with radio parameter $j$,
and $\sum\nolimits_{j = 1}^N {{p_{ij}}} = 1$.
$N=2^T$ is the number of radio parameters.
In Fig. \ref{fig_system_model}, there are $8$ radio parameters and the radio parameter of mesh 2 is 0.

The radio parameter error (RPE) of mesh $i$ is defined as
\begin{equation}\label{eq_pe_definition}
{p_{e,i}} = 1 - \mathop {\max {p_{ij}}}\limits_j,
\end{equation}
and the RPE of the entire region is defined as
\begin{equation}
{p_e} = \sum\limits_{i = 1}^M {{\alpha _i}{p_{e,i}}}
\end{equation}
where $M$ is the number of meshes.
The RPE is not a continuous and smooth function of $p_{ij}$, thus we define the geographic entropy (GE) for convenience.
The geographic entropy of a mesh is defined as the corresponding uncertainty of the radio environment in this mesh.
In Fig. \ref{fig_system_model}, we are more certain about the radio environment in mesh \#1 than that in mesh \#2, since the radio environment
in mesh \#2 is more composite.
Similar to the Shannon entropy \cite{Cover}, the geographic entropy of mesh $i$ is defined as
\begin{equation}\label{GEmesh}
{H_i} =  - \sum\limits_{j = 1}^N {{p_{ij}}\log {p_{ij}}},
\end{equation}
and the geographic entropy of the entire region is defined as
\begin{equation}
H = \sum\limits_{i = 1}^M {{\alpha _i}{H_i}}  =  - \sum\limits_{i = 1}^M {{\alpha
_i}\sum\limits_{n = 1}^N {{p_{ij}}\log {p_{ij}}} }
\end{equation}
where $\alpha _i$ is the area fraction of mesh $i$ compared to the area of the entire region.
For a regular mesh division, such as Fig. \ref{fig_system_model}, ${\alpha _i} = \frac{1}{M}$ and
\begin{equation}\label{eq_entropy_rate}
H = \frac{1}{M}\sum\limits_{i = 1}^M {{H_i}}.
\end{equation}

\subsection{RPE and GE properties}
In this section, we investigate the geographic entropy and radio parameter error,
and the relationship between them.
\begin{Theo}
\label{th_jiangshang}
The geographic entropy of the entire region is ${{\rm O}}({\frac{1}{{\sqrt M }}} ) \to 0$,
where $M$ is the number of meshes.
\end{Theo}
\begin{proof}
The meshes with an impure radio environment are distributed along the network boundaries (shown as a solid curve in Fig. \ref{fig_packing}).
Denote the length of all these boundaries as $\xi$, the length of a mesh edge as $\varepsilon$, and the length of the region edge as $L$.
Then we have $M = {( {\frac{L}{\varepsilon }} )^2}$, and
the number of meshes with impure radio environment is upper bounded by
\begin{equation}\label{eq_K_upper_bound}
K \le \frac{{2\xi \sqrt 2 \varepsilon }}{{{\varepsilon ^2}}} = \frac{{2\sqrt 2 \xi}}{\varepsilon }.
\end{equation}
This result is obtained by considering the corresponding packing problem along the boundary, which is
shown as a solid curve in Fig. \ref{fig_packing}.
Moving each point on this line in the two normal directions a distance $\sqrt 2 \varepsilon$ gives the two dotted lines.
The area between these lines is $2\xi \sqrt 2 \varepsilon$.
All the meshes with an impure radio environment are located between these dotted lines, so
an upper bound on $K$ is $2\xi \sqrt 2 \varepsilon$ divided by the area of a mesh.
An upper bound on the geographic entropy is then
\begin{equation}
H \le \frac{1}{M}K\log N \le \frac{1}{M}\frac{{2\sqrt 2 \xi }}{\varepsilon }\log N =
\frac{1}{{\sqrt M }}\frac{{2\sqrt 2 \xi \log N}}{{\sqrt S }},
\end{equation}
so that ${\rm O} \left( {\frac{1}{{\sqrt M }}} \right) \to 0$ is an upper bound on $H$.
\end{proof}

\begin{figure}[!t]
\centering
\includegraphics[width=0.35\textwidth]{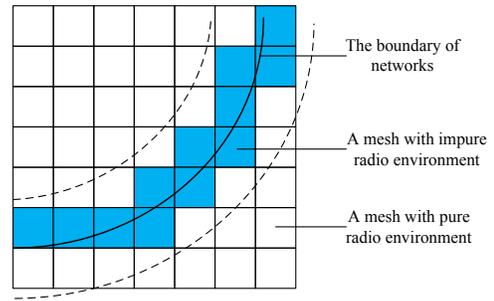}
\caption{The boundary used to determine an upper bound on $K$.} \label{fig_packing}
\end{figure}

\begin{Theo}\label{TH_pe_upper}
The RPE of the entire region is ${{\rm O}}\left( {\frac{1}{{\sqrt M }}} \right) \to 0$, where $M$ is the number of meshes.
\end{Theo}
\begin{proof}
From $1 - {p_{e,i}} = \mathop {\max {p_{ij}}}\limits_j  \ge \frac{1}{N}$, we have that
\begin{equation}\label{eq_pei_upper_bound}
{p_{e,i}} \le 1 - \frac{1}{N}.
\end{equation}
The RPE of the entire region is then upper bounded by
\begin{equation}\label{eq_RPE_upper_bound}
{p_e} \le \frac{1}{M}K\left( {1 - \frac{1}{N}} \right) \le \frac{1}{{\sqrt M }}\frac{{2\sqrt 2 \xi L}}{S}\left( {1 - \frac{1}{N}} \right)
\end{equation}
which gives the required result.
\end{proof}

Theorems \ref{th_jiangshang} and \ref{TH_pe_upper} show the scaling of the GE and RPE as a function of $M$.
The relationship between these parameters is given in the following theorem.
\begin{Theo}\label{th_relation}
The geographic entropy as a function of the radio parameter error is upper bounded by
\begin{equation}\label{origion}
H \le H(p_e) + {p_e}\log \left| {N - 1} \right| \buildrel \Delta \over = \psi ({p_e})
\end{equation}
\end{Theo}

\begin{proof}
According to Fano's inequality, we have
\begin{equation}
{H_i} \le H({p_{e,i}}) + {p_{e,i}}\log \left| {N - 1} \right| \buildrel \Delta \over = \psi ({p_{e,i}})
\end{equation}
where $H_i$ is the entropy of mesh $i$.
Taking the sum gives
\begin{equation}
H = \frac{1}{M}\sum\limits_{i = 1}^M {{H_i}}  \le \frac{1}{M}\sum\limits_{i = 1}^M {\psi ({p_{e,i}})} \mathop  \le \limits^{(a)} \psi \left( {\frac{1}{M}\sum\limits_{i = 1}^M {{p_{e,i}}} } \right),
\end{equation}
where $(a)$ is due to Jensen's inequality from the concavity of $\psi (x)$.
Using ${p_e} = \frac{1}{M}\sum\limits_{i = 1}^M {{p_{e,i}}}$, the proof is complete.
\end{proof}
The following information theoretic lower bound on the entropy as a function of the error probability was given by Feder and Merhav \cite{Meir}.
\begin{Lemm}[\cite{Meir}]\label{lemma_meir}
A lower bound on the entropy $h$ as a function of the error probability $\pi$ is given by $h \ge {\phi}(\pi)$ where
\begin{equation}
\begin{array}{l}
 {\phi}(\pi ) \\
  = \left\{ {\begin{array}{*{20}{c}}
   {{a_1}\pi  + {b_1}} \hfill & {0 \le \pi  \le \frac{1}{2}} \hfill  \\
   {{a_2}(\pi  - \frac{1}{2}) + {b_2}} \hfill & {\frac{1}{2} \le \pi  \le \frac{2}{3}} \hfill  \\
    \vdots  \hfill &  \vdots  \hfill  \\
   {{a_i}(\pi  - \frac{{i - 1}}{i}) + {b_i}} \hfill & {\frac{{i - 1}}{i} \le \pi  \le \frac{i}{{i + 1}}} \hfill  \\
    \vdots  \hfill &  \vdots  \hfill  \\
   {{a_{N - 1}}(\pi  - \frac{{N - 2}}{{N - 1}}) + {b_{N - 1}}} \hfill & {\frac{{N - 2}}{{N - 1}} \le \pi  \le \frac{{N - 1}}{N}} \hfill  \\
\end{array}} \right. \\
 \end{array}
\end{equation}
with ${a_i} = i(i + 1)\log \left( {\frac{{i + 1}}{i}} \right)$ and ${b_i} = \log i$.
\end{Lemm}
From \cite{Meir}, ${\phi}(\pi )$ is an monotone increasing and convex function of $\pi$
(see Fig. 1 in \cite{Meir}, where $\phi^*$ corresponds to $\phi$ in this paper).
Based on Lemma \ref{lemma_meir}, we have the following theorem.
\begin{Theo}\label{th_relation_lowerbound}
A lower bound on the geographic entropy of the entire region as a function of the radio parameter error $p_e$ is given by
\begin{equation}\label{relation_lowerbound}
H \ge {\phi}({p_e}).
\end{equation}
\end{Theo}

\begin{proof}
In Lemma \ref{lemma_meir}, let $h = H_i$ and $\pi = p_{e,i}$ so that
\begin{equation}
{H_i} \ge \phi ({p_{e,i}}).
\end{equation}
The geographic entropy is then
\begin{equation}
H = \frac{1}{M}\sum\limits_{i = 1}^M {{H_i}}  \ge \frac{1}{M}\sum\limits_{i = 1}^M {\phi ({p_{e,i}})}.
\end{equation}
Since $\phi (p_{e,i} )$ is a convex function of $p_{e,i}$, a lower bound on $H$ is given by
\begin{equation}
H \ge \frac{1}{M}\sum\limits_{i = 1}^M {{\phi }({p_{e,i}})} \mathop  \ge \limits^{(b)} {\phi }(\frac{1}{M}\sum\limits_{i = 1}^M {{p_{e.i}}} ) = {\phi }({p_e})
\end{equation}
where $(b)$ is due to Jensen's inequality.
\end{proof}
Combining Theorems \ref{th_relation} and \ref{th_relation_lowerbound}, we have
\begin{equation}
\begin{array}{*{20}{c}}
   {\phi \left( {{p_e}} \right) \le H \le \psi \left( {{p_e}} \right)}  \\
   {{\psi ^{ - 1}}\left( {{H}} \right) \le {p_e} \le {\phi ^{ - 1}}\left( {{H}} \right)}  \\
\end{array}
\end{equation}
Thus the geographic entropy is related to the radio parameter error,
and \emph{an increase of (a reduction of) the geographic entropy may increase (reduce) the radio parameter error}.

The mesh configuration will also affect geographic entropy and the RPE, as shown in the following theorem.
\begin{Theo}
If any two meshes are fused, the entropy of the entire region will not decrease.
\end{Theo}

\begin{proof}
Assume meshes 1 and 2 are fused.
The radio parameter distribution of mesh $i$ is ${p_{i1}},{p_{i2}}, \ldots ,{p_{iN}}$, and the area of mesh $i$ is $s_i$.
The radio parameter distribution of the fused mesh is
\begin{equation}
\begin{aligned}
& \{ {{p'}_1},{{p'}_2}, \cdots ,{{p'}_N}\}  =  \\
& \left\{ \frac{{{s_1}{p_{11}} + {s_2}{p_{21}}}}{{{s_1} + {s_2}}},\frac{{{s_1}{p_{12}} +
{s_2}{p_{22}}}}{{{s_1} + {s_2}}}, \cdots ,\frac{{{s_1}{p_{1N}} + {s_2}{p_{2N}}}}{{{s_1} +
{s_2}}}\right\}
\end{aligned}
\end{equation}
As the entropy is concave \cite{Cover}, we have
\begin{equation}
\begin{aligned}
& \frac{{{s_1}}}{{{s_1} + {s_2}}}H({p_{11}},{p_{12}}, \ldots ,{p_{1N}}) + \frac{{{s_2}}}{{{s_1} + {s_2}}}H({p_{21}},{p_{22}}, \ldots ,{p_{2N}}) \\
&  \le H\left(\frac{{{s_1}{p_{11}} + {s_2}{p_{21}}}}{{{s_1} + {s_2}}},\frac{{{s_1}{p_{12}} +
{s_2}{p_{22}}}}{{{s_1} + {s_2}}}, \ldots ,\frac{{{s_1}{p_{1N}} + {s_2}{p_{2N}}}}{{{s_1} + {s_2}}}\right),
\end{aligned}
\end{equation}
so that
\begin{equation}\label{budengshi}
{s_1}{H_1} + {s_2}{H_2} \le ({s_1} + {s_2})H({{p'}_1},{{p'}_2}, \ldots, {{p'}_N}).
\end{equation}
Before fusion, the entropy of the entire region is
\begin{equation}
H = \frac{{{s_1}{H_1} + {s_2}{H_2}}}{S} + \frac{{\sum\limits_{i = 3}^M {{H_i}} }}{S}
\end{equation}
where $S$ is the region area.
After fusion, this entropy is
\begin{equation}
H' = \frac{{({s_1} + {s_2})H({{p'}_1},{{p'}_2}, \ldots, {{p'}_N})}}{S} +
\frac{{\sum\limits_{i = 3}^M {{H_i}} }}{S}.
\end{equation}
From (\ref{budengshi}), $H \le H'$, so the entropy of the region is not decreased after fusion.
\end{proof}

Duality provides the following theorem.
\begin{Theo}\label{th_mesh_division_reduce_entropy}
Any mesh division operation will not increase the entropy or the radio parameter error of the region.
\end{Theo}

Theorem \ref{th_mesh_division_reduce_entropy} shows that if some meshes are divided into smaller meshes
(such as the meshes with composite radio propagation environment), and
the one-mesh-one-sensor scheme is adopted, then the geographic entropy as well as the RPE can be reduced.
We next examine the tradeoff between the number of sensors and the REM accuracy.

\section{REM Construction Tradeoffs}

In section III, we show that the number of meshes (measurements, sensors) impact the geographic entropy as will as the RPE. In this section, we investigate the tradeoff between the number of sensors and the REM accuracy with a near precise result.

\newcounter{mytempeqncnt}
\begin{figure*}[ht]
\normalsize

\begin{equation}\label{eq_expectation_of_length_in_mesh}
E[{\xi _i}] = \int_0^{\frac{\pi }{4}} {\left( {\int_0^{\sin \theta } {\left( {x\tan \theta  + x\cot \theta } \right){f_X}(x)dx}  + \int_{\sin \theta }^{\frac{{\sqrt 2 L\sin \left( {\theta  + \frac{\pi }{4}} \right)}}{2}} {\frac{1}{{\cos \theta }}{f_X}(x)dx} } \right){f_\Theta }(\theta )d\theta }
\end{equation}

\begin{equation}\label{eq_expectation_of_RPE}
E[{p_{e,i}}] = \int_0^{\frac{\pi }{4}} {\left( {\int_0^{\sin \theta } {\frac{1}{2}\frac{{{x^2}}}{{\sin \theta \cos \theta }}{f_X}(x)dx}  + \int_{\sin \theta }^{\frac{{\sqrt 2 L\sin \left( {\theta  + \frac{\pi }{4}} \right)}}{2}} {\left( {\frac{x}{{\cos \theta }} - \frac{{\tan \theta }}{2}} \right){f_X}(x)dx} } \right){f_\Theta }(\theta )d\theta }
\end{equation}
\hrulefill \vspace*{4pt}
\end{figure*}

\subsection{One-mesh-one-sensor}
If each mesh contains one sensor, the number of sensors equals the number of meshes.
From Theorem \ref{th_relation}, we have
\begin{equation}
p_e \ge {\psi ^{ - 1}}(H),
\end{equation}
which implies that $p_e$ is lower bounded by a function of entropy ${\psi ^{ - 1}}(H)$.
If ${\psi ^{ - 1}}(H) \ne 0$, then the sensing error can never be reduced to $0$.
To reduce the probability of error requires that ${p_e} \le \beta$.
From Theorem \ref{TH_pe_upper}, we have
\begin{equation}\label{eq_lower_bound_sensors}
\begin{aligned}
 & \frac{1}{{\sqrt M }}\frac{{2\sqrt 2 \xi L}}{S}( {1 - \frac{1}{N}}) \le \beta  \Rightarrow M \ge ( {\frac{{2\sqrt 2 \xi (N - 1)}}{{L N\beta }}} )^2 \buildrel \Delta \over = {M_1}
\end{aligned}
\end{equation}
However, because the upper bounds on $p_{e,i}$ and $K$ in Theorem \ref{TH_pe_upper} are loose,
the bound in (\ref{eq_lower_bound_sensors}) is also loose.
Therefore, we use a probability model to obtain near accurate estimates of $K$ and $p_{e,i}$ as follow

\begin{Theo}
The radio parameter error as a function of $M$ is
\begin{equation}
{p_e} = \kappa \frac{1}{{\sqrt M }}
\end{equation}
where
\begin{equation}
\kappa = \frac{{\pi  + \ln 64}}{{12\pi }}\frac{{\pi \xi }}{{ - 4\sqrt 2 {{\tanh }^{ - 1}}(1 - \sqrt 2 )L}}
\end{equation}
which is a constant determined by the length of the boundaries of all networks $\xi$ and the length of entire area's edge $L$.
\end{Theo}

\begin{proof}
Fig. \ref{fig_boundary_in_a_mesh} illustrates the boundary of network in a unit mesh $i$ with an impure radio environment.
This boundary can be approximated by a line when $M$ is large.
We ignore the situation where the boundaries of multiple networks cross the mesh, as the probability of this occurring is low when $M$ is large.
The parameters $x$ and $\theta$ determine a line in Fig. \ref{fig_boundary_in_a_mesh}, where $x$ is the distance between vertex $A$ and the
boundary, and $\theta$ is the angle between this line and horizontal line.
Both $x$ and $\theta$ are random variables with probability density functions (PDFs)
\begin{figure}[!t]
\centering
\includegraphics[width=0.4\textwidth]{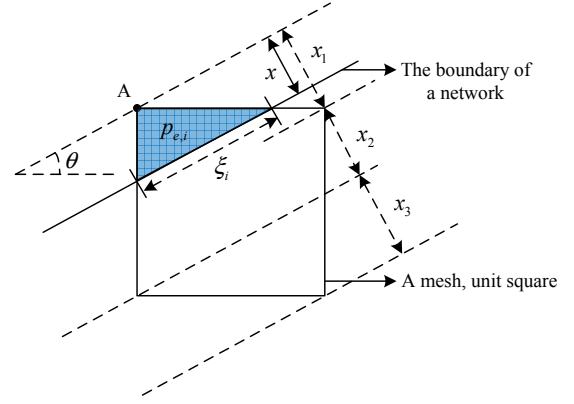}
\caption{A boundary of network cut a mesh.} \label{fig_boundary_in_a_mesh}
\end{figure}
\begin{equation}
{f_\Theta }(\theta ) = \frac{4}{\pi },0 \le \theta  \le \frac{\pi }{4}
\end{equation}
and
\begin{equation}
{f_X}(x) = \frac{2}{{\sqrt 2 L\sin \left( {\theta  + \frac{\pi }{4}} \right)}},0 \le x \le \frac{{\sqrt 2 L\sin \left( {\theta  + \frac{\pi }{4}} \right)}}{2}
\end{equation}
The length of the boundary of network in the mesh is
\begin{equation}
{\xi _i} = \left\{ {\begin{array}{*{20}{c}}
   {x\tan \theta  + x\cot \theta } & {x \le {x_1}}  \\
   {\frac{1}{{\cos \theta }}} & {{x_1} < x < {x_1} + \frac{{{x_2}}}{2}}  \\
\end{array}} \right.
\end{equation}

The radio parameter error is
\begin{equation}
{p_{e,i}} = \left\{ {\begin{array}{*{20}{c}}
   {\frac{{{x^2}}}{{\sin 2\theta }}} & {x \le {x_1}}  \\
   {\frac{x}{{\cos \theta }} - \frac{{\tan \theta }}{2}} & {{x_1} < x < {x_1} + \frac{{{x_2}}}{2}}  \\
\end{array}} \right.
\end{equation}
where $x_1$, $x_2$ and $x_3$ are as shown in Fig. \ref{fig_boundary_in_a_mesh},
with values
\begin{equation}
{x_1} = {x_3} = \sin \theta
\end{equation}
\begin{equation}
{x_2} = {\left[ {\sqrt 2 \sin \left( {\theta  + \frac{\pi }{4}} \right) - 2\sin \theta } \right]^ + }
\end{equation}
with ${[*]^ + } = \max \{ 0,*\}$.
The expected values of $\xi_i$ and $p_{e,i}$ are shown in (\ref{eq_expectation_of_length_in_mesh}) and (\ref{eq_expectation_of_RPE})
at the top of the next page, and the corresponding closed form expressions are
\begin{equation}\label{eq_expectation_xi_i}
E[{\xi _i}] =  - \frac{{4\sqrt 2 {{\tanh }^{ - 1}}(1 - \sqrt 2 )}}{\pi } \cong 0.7935
\end{equation}
and
\begin{equation}\label{eq_expectation_pei}
E[{p_{e,i}}] = \frac{{\pi  + \ln (64)}}{{12\pi }} \cong 0.1937
\end{equation}
where ${\tanh ^{ - 1}}(z)$ is the inverse hyperbolic function defined as ${\tanh ^{ - 1}}(z) = \frac{1}{2}\ln \frac{{1 + z}}{{1 - z}}$.
We determine the value of $K$, i.e., the number of meshes with an impure radio environment, using (\ref{eq_expectation_xi_i}).
If the first $K$ meshes have an impure radio environment, then
\begin{equation}
E[{\xi _i}]\mathop  = \limits^{(a)} \frac{1}{K}\sum\limits_{i = 1}^K {{\xi _i}} \mathop  = \limits^{(b)} \frac{1}{K}\xi
\end{equation}
where $(a)$ is due to the weak Law of Large Numbers (LLN), and $(b)$ is from the fact that the meshes with an impure radio environment cover all the boundaries of networks.
The value of $K$ can be estimated as
\begin{equation}\label{eq_estimation_K}
K = \frac{\xi }{{E[{\xi _i}]}} = \frac{{\pi \xi }}{{ - 4\sqrt 2 {{\tanh }^{ - 1}}(1 - \sqrt 2 )\varepsilon }}
\end{equation}

If the first $K$ meshes have an impure radio environment, then the RPE of the entire region is
\begin{equation}\label{eq_estimation_PRE_whole_region}
{p_e} = \frac{1}{M}\sum\limits_{i = 1}^K {{p_{e,i}}} \mathop  = \limits^{(c)} \frac{K}{M}E[{p_{e,i}}]
\end{equation}
where $(c)$ is due to the LLN.
Substituting the value of $K$ from (\ref{eq_estimation_K}) and the value of $E[{p_{e,i}}]$ from (\ref{eq_expectation_pei})
in (\ref{eq_estimation_PRE_whole_region_final}) gives
\begin{equation}\label{eq_estimation_PRE_whole_region_final}
{p_e} = \frac{1}{{\sqrt M }}\frac{{\pi  + \ln 64}}{{12\pi }}\frac{{\pi \xi }}{{ - 4\sqrt 2 {{\tanh }^{ - 1}}(1 - \sqrt 2 )L}}
\end{equation}
where $\xi$ is the length of all the boundaries of networks, $L$ is the length of the edges of the entire region, and $N$ is the number of radio parameters.
This is a near precise estimate.
\end{proof}

Note that this confirms the result ${p_e} = \Theta \left( {\frac{1}{{\sqrt M }}} \right)$.
Similar to the derivation of (\ref{eq_lower_bound_sensors}), the number of sensors
can be obtained using the more precise estimate of $p_e$ in (\ref{eq_estimation_PRE_whole_region_final}).
From the bound ${p_e} \le \beta$, we have
\begin{equation}\label{eq_M2}
M \ge {\left( {\frac{{(\pi  + \ln 64)\xi }}{{12\left( { - 4\sqrt 2 {{\tanh }^{ - 1}}(1 - \sqrt 2 )} \right)L\beta }}} \right)^2} \buildrel \Delta \over = {M_2}
\end{equation}

\subsection{Random sensor deployment}

Randomly deploying sensors over the entire region is more realistic than one sensor in each mesh.
Suppose there are $J$ sensors and the region is divided into $M$ meshes.
With a uniform deployment, the probability that a sensor falls into mesh $i$ is $\frac{1}{M} \forall i$.
Thus the probability that there are no sensors in mesh $i$ is
\begin{equation}
{p_0} = {\left( {1 - \frac{1}{M}} \right)^J}.
\end{equation}
If $J = kM$, then $\mathop {\lim }\limits_{M \to \infty } {p_0} = {e^{ - k}}$.
Thus, the number of meshes that have no sensors is
\begin{equation}\label{eq_vacant_meshes}
M{p_0} = \frac{M}{{{e^k}}}
\end{equation}
The radio parameters for meshes without a sensor are randomly chosen, so
the maximum error probability for an empty mesh is still ${1 - \frac{1}{N}}$.
An upper bound on the radio parameter error is then given by
\begin{equation}\label{eq_random_deployment_pe}
p_e^* = \frac{1}{M}\sum\limits_{i = 1}^M {{p_{e,i}}}  \le {p_e} + \frac{1}{M}M{p_0}\left( {1 - \frac{1}{N}} \right)
\end{equation}
where $p_e$ is obtained from (\ref{eq_estimation_PRE_whole_region_final}).
Since $p_e^* \le \beta$, we have
\begin{equation}\label{eq_random_sensor_density}
\begin{aligned}
& M \ge {\left( {\frac{{(\pi  + \ln 64)\xi }}{{12( { - 4\sqrt 2 {{\tanh }^{ - 1}}(1 - \sqrt 2 )} )L( {\beta  - {e^{ - k}}( {1 - \frac{1}{N}} )} )}}} \right)^2}\\
& \buildrel \Delta \over = {M_3}
\end{aligned}
\end{equation}

Note that $M_3 > M_2$, i.e., with random deployment the number of sensors required to achieve the same error probability as
the one-mesh-one-sensor scheme is larger.

Finally, we analyze the situation when a mesh contains more than one sensor.
In this case, the mesh is divided into smaller meshes such that each smaller mesh contains one sensor.
From Theorem \ref{th_mesh_division_reduce_entropy}, any division operation will not reduce the geographic entropy and therefore not reduce the error probability.
Thus (\ref{eq_random_deployment_pe}) is still an upper bound on the RPE and (\ref{eq_random_sensor_density}) is still a lower bound on the number of meshes (sensors)
with random deployment.

\begin{figure}[!t]
\centering
\includegraphics[width=0.5\textwidth]{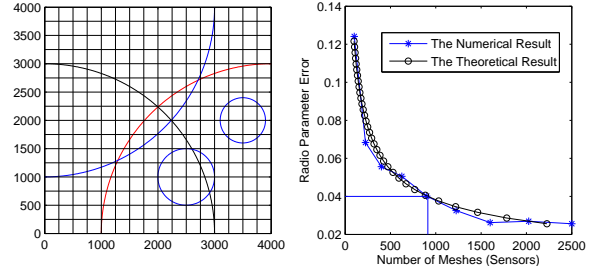}
\caption{The number of sensors vs. the RPE with five networks.} \label{fig_five_nets}
\end{figure}

\begin{figure}[!t]
\centering
\includegraphics[width=0.5\textwidth]{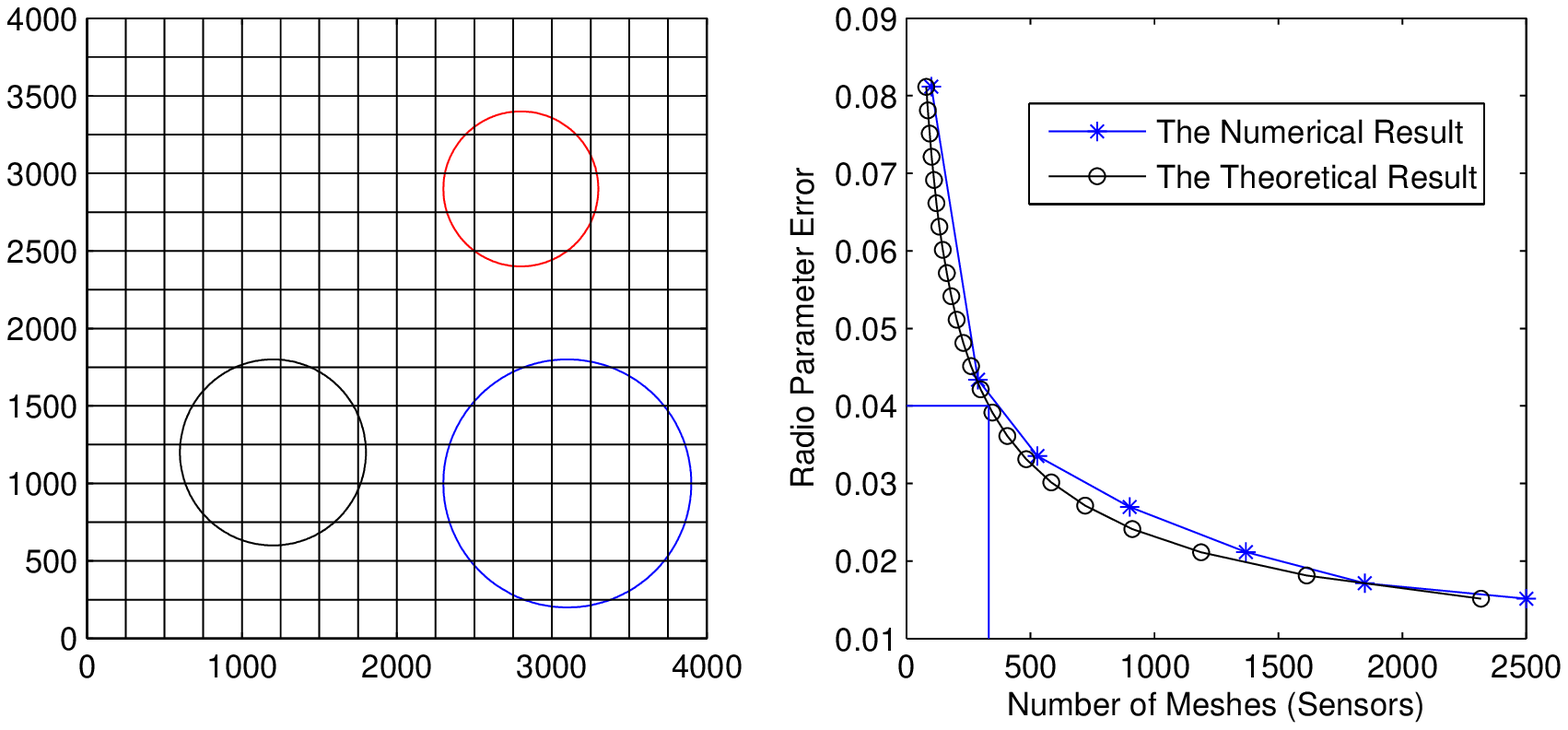}
\caption{The number of sensors vs. the RPE with three networks.} \label{fig_three_nets}
\end{figure}

\section{Numerical Results}

The relationship between the number of meshes (sensors) and the radio parameter error (RPE) in shown in Figs. \ref{fig_five_nets} and \ref{fig_three_nets}
for five and three networks, respectively.
This verifies (\ref{eq_estimation_PRE_whole_region_final}) for the one-mesh-one-sensor scheme.
The RPE is related to the number of sensors $M$ and the length of the network boundaries $\xi$.
As $M$ increases, the RPE decreases, and as $\xi$ increases, the RPE increases.
The value of $\xi$ in Fig. \ref{fig_five_nets} is larger than the corresponding value in Fig. \ref{fig_three_nets}.
Thus to achieve the same RPE (for example, RPE = 0.04), the number of meshes for five networks must be larger than the number with three networks.
Note that the RPE is a convex function of $M$, thus when $M$ is sufficiently large, the RPE improvement is not significant.

From (\ref{eq_M2}) and (\ref{eq_random_sensor_density}), for the same error probability, the number of sensors with the one-mesh-one-sensor scheme
will be smaller that with random sensor deployment.
This is verified by Fig. \ref{fig_sensor_number}.
As $k$ increases, the number of sensors with random deployment approaches the number with the one-mesh-one-sensor scheme.
This results can be obtained from (\ref{eq_random_sensor_density}) as $\mathop {\lim }\limits_{k \to \infty } {M_3} = {M_2}$.

Fig. \ref{fig_picture_REM_construction} provides three examples of REM construction.
In Fig. \ref{fig_picture_REM_construction}($\rm{a_1}$), the REM with the one-mesh-one-sensor scheme has errors along the network boundaries.
This is because the radio environment of the meshes along the boundaries is impure and may contain errors in the measurement results.
The radio parameter error of each mesh is illustrated in Fig. \ref{fig_picture_REM_construction}($\rm{a_2}$).
Figs. \ref{fig_picture_REM_construction}($\rm{b_1}$) and ($\rm{c_1}$) illustrate the REM with random sensor deployment and $k=1$ and $k=2$, respectively.
Note that the results in Fig. \ref{fig_picture_REM_construction}($\rm{c_1}$) are more accurate than in ($\rm{b_1}$).
This is expected since (\ref{eq_random_sensor_density}) indicates that the RPE is a decreasing function of $k$.
Figs. \ref{fig_picture_REM_construction}($\rm{b_2}$) and ($\rm{c_2}$) show the radio parameter error distribution for the REMs in
Figs. \ref{fig_picture_REM_construction}($\rm{b_1}$) and ($\rm{c_1}$). respectively.
A mesh without any sensors is shown in red.
Note that Fig. \ref{fig_picture_REM_construction}($\rm{c_2}$) contains fewer red meshes than Fig. \ref{fig_picture_REM_construction}($\rm{b_2}$).
This confirms (\ref{eq_vacant_meshes}), which indicates that the number of vacant meshes is a decreasing function of $k$.
In Figs. \ref{fig_picture_REM_construction}($\rm{a_2}$), ($\rm{b_2}$) and ($\rm{c_2}$), the meshes that are not red or dark blue
denote meshes that have an incorrect radio parameter.
The number of such meshes is lower in Figs. \ref{fig_picture_REM_construction}($\rm{b_2}$) and ($\rm{c_2}$)
than in Fig. \ref{fig_picture_REM_construction}($\rm{a_2}$).
This indicates that random sensor deployment results in fewer measurement errors, but there are more meshes with no sensors when $k > 1$.

\begin{figure}[!t]
\centering
\includegraphics[width=0.35\textwidth]{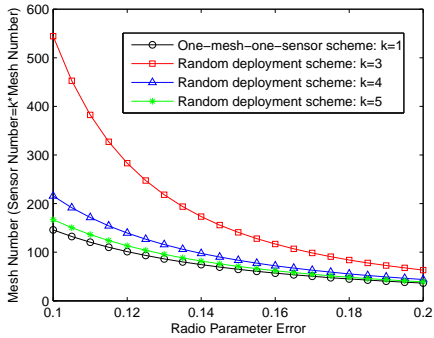}
\caption{The relationship between the number of sensors and the radio parameter error.}
\label{fig_sensor_number}
\end{figure}

\begin{figure}
\centering
\includegraphics[width=0.5\textwidth]{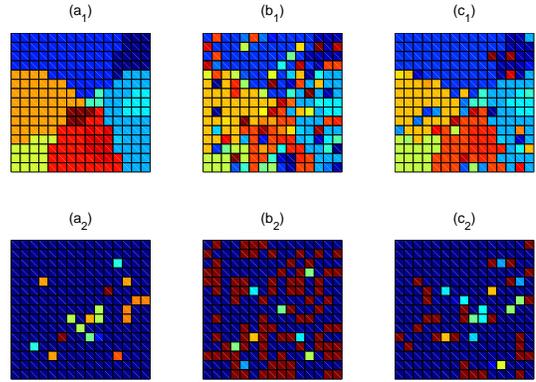}
\caption{REM construction results and the corresponding RPE values for the networks in Fig. \ref{fig_five_nets};
($\rm{a_1}$) corresponds to the one-mesh-one-sensor scheme with $16 \times 16$ meshes (sensors);
($\rm{b_1}$) corresponds to random sensor deployment with $16 \times 16$ meshes and $k=1$;
and ($\rm{c_1}$) corresponds to random sensor deployment scheme with $16 \times 16$ meshes and $k=2$.
The RPE values for ($\rm{a_1}$), ($\rm{b_1}$) and ($\rm{c_1}$) are given in ($\rm{a_2}$), ($\rm{b_2}$) and ($\rm{c_2}$), respectively.} \label{fig_picture_REM_construction}
\end{figure}

\section{Conclusion}

In this paper, we have achieved the relationship between the number of sensors and the radio environment map (REM) accuracy.
The concept of geographic entropy is introduced to quantify this relationship.
And the influence of sensor deployment on REM accuracy is examined using information theory techniques.
The results obtained in this paper are applicable not only for the REM, but also for wireless sensor network deployment.


\begin{thebibliography}{99}

\bibitem{CR_Mitola}
J. Mitola,
{\em Cognitive radio: An integrated agent architecture for software defined radio},
Ph.D. dissertation, KTH Royal Inst. of Technol., Stockholm, Sweden, 2000.

\bibitem{CR_Performance_Evaluation}
Y. Zhao, S. Mao, J. O. Neel, and J. H. Reed,
``Performance evaluation of cognitive radios: Metrics, utility functions, and methodology,''
{\em Proc. IEEE}, vol. 97, no. 4, pp. 642--659, Apr. 2009.

\bibitem{REM_Fast_Algorithm}
S. Grimoud, B. Sayrac, S. Ben Jemaa, and E. Moulines, ``An algorithm for fast REM construction,''
Cognitive Radio Oriented Wireless Networks and Communications (CROWNCOM), pp. 251--255, June 2011.

\bibitem{REM_spatial_statistics}
J. Riihij\"{a}rvi, P. M\"{a}h\"{o}nen, M. Petrova, and V. Kolar,
``Enhancing cognitive radios with spatial statistics: From radio environment maps to topology engine,''
Cognitive Radio Oriented Wireless Networks and Communications (CROWNCOM),
pp. 1--6, June 2009.

\bibitem{REM_Heterogeneous_sensor}
V. Atanasovski et al.,
``Constructing radio environment maps with heterogeneous spectrum sensors,''
IEEE Symp. on New Frontiers in Dynamic Spectrum Access Networks,
pp. 660--661, May 2011.

\bibitem{REM_number_sensor_and_REM}
S. Faint, X. O. \"{U}reten, and T. Willink,
``Impact of the number of sensors on the network cost and accuracy of the radio environment map,''
IEEE 23rd Canadian Conference on Electrical and Computer Engineering (CCECE),
pp. 1--5, May 2010.

\bibitem{CPC_Ondemand}
J. Perez-Romero, O. Salient, R. Agusti, and L. Giupponi,
``A novel on-demand cognitive pilot channel enabling dynamic spectrum allocation,''
IEEE International Symposium on New Frontiers in Dynamic Spectrum Access Networks (DySPAN),
pp. 46--54, Apr. 2007.

\bibitem{Mesh}
Z. Wei and Z. Feng,
``A geographically homogeneous mesh grouping scheme for broadcast cognitive pilot channel in heterogeneous wireless networks,''
IEEE GLOBECOM Workshops, pp. 1008--1012, Dec. 2011.

\bibitem{Mesh_SCI}
Z. Feng, Z. Wei, Q. Zhang and P. Zhang . ``Fractal theory based dynamic mesh grouping scheme for efficient cognitive pilot channel design,'' Chinese Science Bulletin, vol. 57, no. 28-29, pp. 3684--3690, Nov. 2012.

\bibitem{Meir}
M. Feder and N. Merhav,
``Relations between entropy and error probability,''
IEEE Transactions on Information Theory, vol. 40, no. 1, pp. 259--266, Jan. 1994.

\bibitem{Cover}
T. Cover and J. Thomas,
{\em Elements of Informantion Theory},
Wiley, New York, 2006.


\end{thebibliography}
\end{document}